\documentclass[11pt,letterpaper]{article}

\usepackage[margin=1in]{geometry}
\usepackage[T1]{fontenc}
\usepackage{amsmath,amsthm}
\usepackage{newtxtext,newtxmath}
\usepackage{microtype}
\usepackage[hyphens]{url}
\usepackage{xurl}
\usepackage{graphicx}
\usepackage[authoryear,round]{natbib}
\usepackage{caption}
\usepackage{booktabs}
\usepackage{float}
\usepackage{hyperref}
\usepackage[nameinlink,noabbrev]{cleveref}

\newcommand{\ArxivAuthorName}{Hiroaki Odahara}
\newcommand{\ArxivAffiliation}{%
Market Design Center, Graduate School of Economics, The University of Tokyo\\
Graduate School of Informatics and Engineering, The University of Electro-Communications}
\newcommand{\ArxivDate}{July 2026}
\newcommand{\ArxivAcknowledgments}{%
This work was supported by JST ERATO Grant Number JPMJER2301, Japan.}

\hypersetup{
  colorlinks=true,
  linkcolor=black,
  citecolor=black,
  urlcolor=black,
  pdfauthor={\ArxivAuthorName},
  pdftitle={Reversing Reserve Logic: Optimal Holdback in Local Allocation under Scalable Entry}
}

\newtheorem{theorem}{Theorem}
\newtheorem{proposition}[theorem]{Proposition}
\newtheorem{lemma}[theorem]{Lemma}
\newtheorem{corollary}[theorem]{Corollary}

\newcommand{\E}{\mathbb{E}}
\newcommand{\Prb}{\mathbb{P}}
\newcommand{\one}{\mathbf{1}}
\newcommand{\PS}{\mathrm{PS}}
\newcommand{\dd}{\,\mathrm{d}}

\title{Reversing Reserve Logic:\\
Optimal Holdback in Local Allocation under Scalable Entry}
\author{\ArxivAuthorName\\[0.35em]
\small \ArxivAffiliation}
\date{\ArxivDate}

\begin{document}
\maketitle

\begin{abstract}
Scarce opportunities such as concert tickets and accelerator time may be
contested by automated participants that can create accounts and sustain
commitments beyond the reach of commitment-limited intended users. When
account counts are untrusted, we study anonymous screening rules that ignore
them, cap retained burdens, use only an account's commitment and strongest
rival, and do not reassign after rejecting the leader. Within this class, we
characterize the rule maximizing intended users' expected utility when they
commit fully and a scalable entrant stays out. The optimum refunds and
allocates at low congestion, retains and allocates at intermediate congestion,
and retains while withholding allocation from an otherwise eligible leader
when the strongest rival lies in the upper tail. Unlike a conventional
reserve, which rejects a low leading bid, this rule treats an unusually strong
rival as evidence of entrant imitation. A direct dual certificate proves class
optimality; a benchmark shows that upper-tail holdback can raise intended-user
surplus before it is necessary to support non-entry. The rule supports an
equilibrium with full commitment and entrant non-entry.
\end{abstract}

\section{Introduction}
\label{sec:intro}

Auctions may leave a scarce opportunity unallocated despite a unique highest
bid. A familiar revenue-maximizing explanation is that the highest bid fails
to meet a reserve price \citep{Myerson1981}. We reverse this screening logic
for an operator prioritizing intended-user access over current receipts. In an
open-identity environment, account counts are
unreliable, while a scalable entrant with lower direct-use value may sustain
commitments beyond the reach of commitment-limited intended users. Holding the
leading commitment fixed, an unusually strong runner-up is more likely when
the leader is the entrant than when it is an intended user. Within this class,
the optimum may retain the otherwise eligible
leader's posted commitment while withholding allocation. The leader is
rejected not because its commitment is too low, but because the runner-up is
too high.

Intended users may include young fans whose ability to post a large deposit
need not reflect attendance value, or small research teams that value AI
accelerator time despite tight liquidity or authorization limits. Automated
participants and commercial aggregators can replicate accounts and coordinate
commitments, making purchase limits fragile
\citep{Courty2019,USCongress2016BOTS}. In AI infrastructure, an account is an
autonomous client instance, a commitment is a verifiable reservation or
liquidity authorization, and the opportunity is an expiring compute or
priority slot. Ticketing and AI infrastructure are motivating applications,
not calibrations.

The operator asks each account to post a verifiable commitment before assigning
one current opportunity. Treatment may depend on the account's commitment and
its strongest observed rival, but not on account count or identity. Count
independence is an operational robustness requirement, not a consequence of
knowing the intended-user population. The announced rule is single-pass: after
all submissions, the operator resolves the opportunity once and does not
reassign it after rejecting the unique leader. For that leader, the strongest
rival is the runner-up commitment. The rule uses it as evidence about the
leader, not to select its owner as a fallback. Holdback denotes deliberate
current nonallocation, without fallback or a promise of later service.

Retention burdens intended users but preserves current allocation value while
reducing entrant payoff. The pointwise linear program therefore uses it before
holdback, producing three bands: refund and allocate at low congestion, retain
and allocate at intermediate congestion, and retain while withholding only in
the upper tail. A leader-always-allocate comparison sharpens the interpretation.
Holdback can strictly improve intended-user surplus even while a rule that
always allocates to the unique leader can still support entrant non-entry, so
optimality precedes necessity.

A continuous linear program describes the design problem. A pointwise
weak-duality certificate produces two congestion cutoffs and the exact class
optimum for any continuous commitment-capacity distribution satisfying the
maintained assumptions. Under that rule, every intended user with positive
capacity uniquely prefers to commit fully when all other intended users do so
and the entrant stays out. The entrant is indifferent between staying out and
participating, so the rule supports non-entry without uniquely selecting it.
An external entrant-specific account-creation cost makes non-entry uniquely
optimal without changing the rule.

A natural objection is that requiring full commitment may itself create the
result. Intended users who anticipate upper-tail holdback might shade their
commitments, pool at a common interior cap, or mix over lower actions. A
full-commitment reduction, together with the actionwise upper bound, shows
that none of these responses improves joint rule-and-policy design within the
maintained class. Given a feasible local rule and a common, independently
randomized, capacity-feasible one-account policy under which entrant non-entry
is a best response, the reduction constructs a full-commitment implementation
while preserving the local information structure, burden cap, unconditional
joint distribution of assigned allocation probabilities and retained burdens,
aggregate intended-user surplus, and non-entry.

\paragraph{Contributions.}
The paper makes three contributions. First, it characterizes the exact
class-optimal three-band rule and certifies it by pointwise weak duality. The
rule refunds and allocates at low congestion, retains and allocates at
intermediate congestion, and retains while withholding in the upper tail.
Second, it separates optimal holdback from necessary holdback. In a two-user
uniform benchmark, upper-tail nonallocation strictly raises intended-user
surplus even though a leader-always-allocate rule can still support entrant
non-entry. Third, a full-commitment reduction extends the same value bound to
common, independently randomized, capacity-feasible one-account policies under
which entrant non-entry is a best response.

\paragraph{Scope and interpretation.}
The exact optimality claim is conditional on anonymous, count-independent,
strongest-rival, single-pass rules. The policy comparison covers common,
independently randomized one-account behavior subject to individual capacity
limits. The optimization excludes asymmetric or publicly correlated behavior, general
multiple-account policies, direct posting costs, and type-dependent values.
The operator is modeled as prioritizing intended-user access rather than
current receipts: its objective is intended users' allocation value net of
retained burdens, with no weight on entrant value. It is not a total-welfare
or within-group-equity objective. This access objective can represent mission,
reputation, or long-run participation concerns that make service to intended
users more important than current receipts; those motives are not modeled
separately. Allowing entrant service would change both the target behavior and
the evaluated objective and lies outside the optimality theorem.

\paragraph{Relation to prior work.}
Money burning and costly-signal mechanisms screen agents through wasteful
payments or ordeals
\citep{HartlineRoughgarden2008,Condorelli2012,ChakravartyKaplan2013,
NicholsZeckhauser1982,Condorelli2013}. Here the diagnostic state is observed
congestion, and the design question is how allocation and burdens should be
ordered across that state. Queues, contests, and financially constrained
allocation provide related strategic environments
\citep{Naor1969,HoltSherman1982,Leshno2022,MoldovanuSela2001,Siegel2009,
PaiVohra2014}.

False-name manipulation in anonymous systems has been studied in auctions,
costly voting, and Bayesian Vickrey--Clarke--Groves mechanisms
\citep{YokooSakuraiMatsubara2004,WagmanConitzer2008,GafniLaviTennenholtz2020}.
\citet{PanEtAl2026Sybil} characterize mechanisms under non-wastefulness,
incentive compatibility, symmetry, and ex post Sybil-proofness. Our class
permits holdback and obtains only a designated-profile result for unilateral
identity splitting under additive retained burdens. \citet{MazorraDellaPenna2023}
instead study identity costs and commitment to future Sybil strategies.
Seller-side shill bidding is another distinct manipulation
\citep{KomoKominersRoughgarden2025}.

Ticket markets and fan-allocation platforms motivate the application
\citep{LeslieSorensen2014,Courty2019,BudishBhave2023,TangZengZuo2017}.
Within AI and multiagent systems, related work studies strategic
shared-resource allocation, Bayesian scarce-facility mechanisms,
manipulation-aware testing, and adversarial auditing
\citep{FriedmanEtAl2019,AuricchioZhang2026,QiuShan2026,DasYuZhang2026}.
These models optimize resource allocation, facility placement, test
sequencing, or verification. Our distinct combination is untrusted account
counts, bounded commitments, and scalable entry. Within it, we characterize
when a count-independent local rule uses congestion-contingent current
nonallocation.

\section{Model and Local Rule Class}
\label{sec:model}

A single indivisible access opportunity is available, with the commitment cap
normalized to one. There are $K\geq2$
intended users, each commitment-limited. Each has the same direct-use value $v_C>1$ and independently
drawn commitment capacity $\kappa_i\sim F$ on $[0,1]$. The CDF $F$ is strictly
increasing and absolutely continuous, with density $f>0$ almost everywhere.
Intended user $i$ may submit any commitment $a_i\in[0,\kappa_i]$. A scalable entrant
$D$ obtains value $v_D\in(0,v_C)$ from obtaining the opportunity for
operational or commercial use and can choose any commitment in $[0,1]$.
Throughout, subscript $C$ labels intended-user quantities and $D$ labels the
scalable entrant.

Figure~\ref{fig:three-regimes} previews the class-optimal three-band rule. Its
green region displays the central design result: the rule rejects an otherwise
eligible leader when the strongest rival commitment reaches the upper tail.
The middle panel adds that this targeted holdback can be optimal even though a
leader-always-allocate rule can still support entrant non-entry.

\begin{figure*}[t]
\centering
\includegraphics[width=\textwidth]{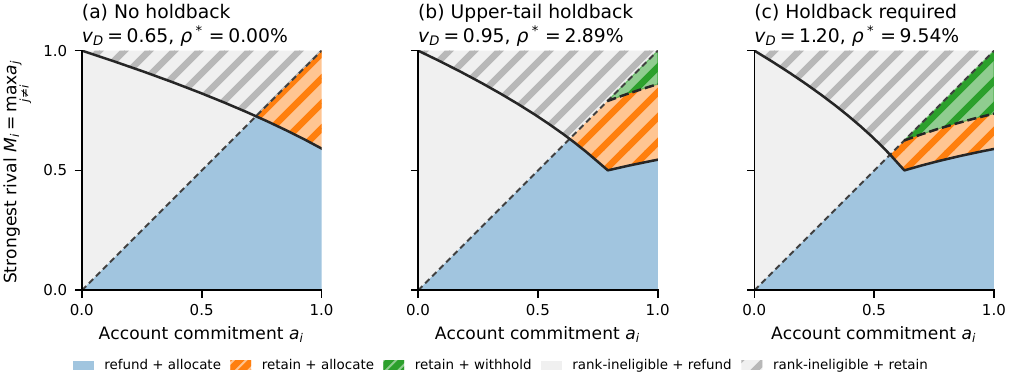}
\caption{Upper-tail holdback in the class-optimal rule. State maps for $K=2$,
$F(a)=a$, and $v_C=1.5$. At $v_D=0.65$, $0.95$, and $1.20$,
respectively, holdback is absent, surplus-improving but unnecessary for
supporting non-entry, and necessary for such support. The panel titles report
the ex ante holdback probability $\rho^*$. Colors indicate refund/allocate,
retain/allocate, and retain/withhold. Gray denotes states in which the strongest
rival weakly exceeds the account's own commitment, making the account
rank-ineligible. Diagonal bands always mean full retention: plain gray is
refunded, whereas striped gray is retained. The solid curve separates refund
from full retention; the off-diagonal dashed curve separates allocation from
withholding. Conditional on a submitted profile, the rule is deterministic
almost everywhere: green denotes targeted upper-tail holdback, not a
state-independent failure lottery.}
\label{fig:three-regimes}
\end{figure*}

The population composition and the primitives $K$, $v_C$, $v_D$, and $F$ are common
knowledge. The private information is an intended user's commitment capacity,
not her value. Account ownership and controller role are not observable or
verifiable: after seeing the submitted commitments, the operator cannot tell
whether an account is controlled by an intended user or by $D$, or link
several accounts to a common controller. The rule therefore cannot condition
on an intended-user/entrant label. Commitment capacity is conceptually
distinct from use value and caps only enforceable exposure. An intended user
may therefore value access highly while having little enforceable liquidity.
Cash liquidity is the leading interpretation, but the same notation also
covers deposits, liquidity authorizations, and similar enforceable burdens.

\paragraph{Timing.}
Each intended user privately observes her capacity. Intended users and the
entrant choose commitments simultaneously without observing others' realized
capacities or commitments; the entrant may instead abstain. The operator then
applies the precommitted rule once. Thus an entrant deviation is a fixed
action, not a policy contingent on realized congestion.

Let $N$ be the set of submitted accounts and
$\boldsymbol a=(a_j)_{j\in N}$ their commitment profile. For each $i\in N$, define
the strongest-rival statistic
$M_i(\boldsymbol a_{-i})=\max_{j\in N\setminus\{i\}}a_j$. A measurable local
mechanism applies the same functions $x$ and $r$ to every account, giving
account $i$ current-round allocation probability $x(a_i,M_i)\in[0,1]$ and retained burden
$r(a_i,M_i)\in[0,a_i]$. Only an account with a strictly highest commitment can
receive the unit, so $x(a_i,M_i)\leq\one\{M_i<a_i\}$. Because $F$ is
continuous, ties have probability zero when intended users submit their
capacities; the enlarged policy comparison below permits atoms and positive
tie probability. The mechanism is
anonymous and account-count-independent: apart from $(a_i,M_i)$, treatment
cannot depend on identity, account labels, or the number of submissions.
We call a measurable local rule pointwise feasible if
$0\leq x(a,m)\leq\one\{m<a\}$ and $0\leq r(a,m)\leq a$ for every
$(a,m)\in[0,1]^2$.

\paragraph{Why the strongest rival?}
Restricting treatment to $(a_i,M_i)$ is stronger than account-count
independence alone. Zero-commitment identities can change submission counts
without retained-burden exposure, whereas raising $M_i$ requires another
account to accept the corresponding enforceable commitment exposure. This is
a manipulation-robustness rationale, not a characterization of every robust
statistic. This is a maintained operational restriction, not a
without-loss reduction.
Corollary~\ref{cor:finite-false-name} later addresses finite identity
splitting; the technical appendix examines the additional but fragile power of exact
count conditioning. In the focal-account calculations below, $a$ denotes a
candidate own action and $m$ a realization of $M_i$. Intended users and the
entrant are risk neutral, and burdens enter utility linearly. When
$x(a,m)=0$ in a rank-eligible state,
the current opportunity is withheld and yields no allocation value; the model
contains no promise of later service or continuation value.

\paragraph{Induced Bayesian game.}
Fix a local rule $(x,r)$. The induced game is specified as follows.
\begin{enumerate}
\item Intended user $i$ has private type $\kappa_i\in[0,1]$ and
chooses an action in $[0,\kappa_i]$. The entrant chooses from
$A_D=\{\varnothing\}\cup[0,1]$, where $\varnothing$ denotes non-entry and
$a=0$ denotes entry with zero commitment.
\item Each submitted account $j\in N$ receives payoff
$v_jx(a_j,M_j(\boldsymbol a_{-j}))-r(a_j,M_j(\boldsymbol a_{-j}))$, where
$v_j=v_C$ for an intended user and $v_j=v_D$ for the entrant. Non-entry gives
the entrant zero.
\item A pure intended-user strategy is a measurable function $s_i$ satisfying
$s_i(\kappa_i)\in[0,\kappa_i]$.
\end{enumerate}

\paragraph{Designated profile.}
The designated profile has $s_i^T(\kappa_i)=\kappa_i$ for every
intended user and $s_D^T=\varnothing$.

\paragraph{Equilibrium convention.}
Equilibrium statements use the typewise convention that a strategy must be a
best response at every intended-user type, rather than only almost everywhere.
The equilibrium-preservation extension of the reduction, proved in the
technical appendix, gives only an almost-everywhere conclusion; the
joint-design bound does not require typewise equilibrium preservation.

Under this profile, an intended user faces the maximum of $K-1$ independent
capacities, whose CDF is $F^{K-1}$; a fixed entrant action faces the maximum of
all $K$ intended-user capacities, whose CDF is $F^K$. Hence an intended user
who chooses $a$ has expected payoff
\begin{equation}
 U_C(a)=\int_0^1\!\left[v_Cx(a,m)-r(a,m)\right]\dd F(m)^{K-1}.
 \label{eq:constrained-utility}
\end{equation}
The rank constraint already sets $x(a,m)=0$ when $m\geq a$. The entrant's
deviation payoff from the same action is
\begin{equation}
 U_D(a)=\int_0^1\!\left[v_Dx(a,m)-r(a,m)\right]\dd F(m)^K.
 \label{eq:entrant-utility}
\end{equation}

\begin{proposition}[Designated-profile equilibrium]
\label{prop:target-equilibrium}
For any pointwise-feasible local rule:
\begin{enumerate}
\item The designated profile is a pure-strategy Bayesian Nash equilibrium if
and only if $U_C$ is nondecreasing and $U_D(a)\leq0$ for every $a\in[0,1]$.
\item Conditional on the other players' designated strategies, full
commitment is the unique best response for an intended user of every type
$\kappa>0$ if
and only if $U_C$ is strictly increasing.
\item If $U_D(a)=0$ for every $a\in[0,1]$, then non-entry and every commitment
in $[0,1]$ are best responses for the entrant. Thus non-entry is supported but
is not uniquely selected.
\end{enumerate}
\end{proposition}
\begin{proof}
An intended user of type $\kappa$ has action set $[0,\kappa]$ and obtains
$U_C(a)$ from action $a$ against the designated opponents. This proves the two
claims for intended users.
The entrant obtains zero from non-entry and $U_D(a)$ from action $a$, which
proves the remaining claims.
\end{proof}

Because $U_C(0)=0$, nondecreasing designated-profile utility also implies Bayesian
interim individual rationality (IR). Ex post IR is not imposed:
a losing account can forfeit its full commitment.

\section{The Optimal Rule: Upper-Tail Holdback}
\label{sec:mechanism}

Let $\mathcal M$ be the set of pointwise-feasible local rules. The designer
solves
\begin{equation}
\begin{aligned}
 (\mathsf P)\quad
 \sup_{(x,r)\in\mathcal M}\quad &K\int_0^1 U_C(a)\dd F(a)\\
 \text{s.t.}\quad &U_D(a)\leq0
     &&\forall a\in[0,1],\\
 &U_C(a)\geq0
     &&\forall a\in[0,1],\\
 &U_C(a')\geq U_C(a)
     &&\forall\,0\leq a\leq a'\leq1.
\end{aligned}
\label{eq:design-lp}
\end{equation}
The objective is intended-user surplus: retained burdens are losses, not
transfers credited to the operator. Revenue maximization is a different
objective. Problem~\eqref{eq:design-lp} is linear in the two state-contingent
functions. The main proof first solves its
actionwise relaxation by weak duality and then verifies that the resulting
utility makes full commitment a strict best response at the designated
profile. The actionwise upper bound omits intended-user incentive and
participation constraints. It is an intentional relaxation rather than a
claim that the designer can prescribe nonequilibrium behavior.

\paragraph{Why the LP generates three bands.}
At a fixed commitment, an intended user faces strongest-rival density
$g_C(m)=(K-1)F(m)^{K-2}f(m)$, whereas a scalable entrant choosing the same
commitment faces $g_D(m)=KF(m)^{K-1}f(m)$. Thus, for Lebesgue-almost every
$m\in(0,1)$,
\begin{equation}
  \frac{g_D(m)}{g_C(m)}=\frac{K}{K-1}F(m).
  \label{eq:mlr}
\end{equation}
This ratio rises strictly with congestion, a monotone likelihood ratio in the
strongest rival. Upper-tail states are therefore relatively more diagnostic
of an entrant choosing that same commitment.

Fix $a>0$ and $q\in(0,1)$, and place the positive multiplier
$\lambda=(K-1)/q$ on the entrant's nonpositive-payoff constraint. Removing a
common positive density factor gives the Lagrangian coefficients
\begin{align}
 \Phi_R(a,m)&=\lambda F(m)-(K-1),\label{eq:phiR}\\
 \Phi_X(a,m)&=(K-1)v_C-\lambda v_DF(m).\label{eq:phiX}
\end{align}
The retention coefficient is positive for $F(m)>q$, whereas the allocation
coefficient is negative only for $F(m)>(v_C/v_D)q>q$. Thus, within
rank-eligible states, the order is refund and allocate, retain and allocate,
then retain and withhold. The value of $q$ is chosen below so that the entrant
constraint binds.

Define
\begin{equation}
 B_K\equiv \frac{v_C^K}{v_D^{K-1}}.
 \label{eq:BK}
\end{equation}
For every own commitment $a>0$, define the retention-cutoff quantile
\begin{equation}
 q_R(a)^K
 =\max\left\{
      1-\frac{v_D}{a}F(a)^K,
      \frac{a}{a+B_K}
    \right\},
 \label{eq:qR}
\end{equation}
and the allocation-cutoff quantile
\begin{equation}
 \bar q_R(a)=\min\left\{1,\frac{v_C}{v_D}q_R(a)\right\}.
 \label{eq:qbar}
\end{equation}
Set $x^*(0,m)=r^*(0,m)=0$. For $a>0$, the candidate rule is
\begin{align}
 r^*(a,m)&=a\one\{F(m)\geq q_R(a)\},
 \label{eq:optimal-retention}\\
 x^*(a,m)&=\one\{F(m)<\min(F(a),\bar q_R(a))\}.
 \label{eq:allocation-rule}
\end{align}
Thus holdback is not a primitive probability: it is the probability mass of
the states in which the allocation cutoff lies below the rank-eligibility
cutoff $F(a)$.

\begin{theorem}[Optimal upper-tail holdback]
\label{thm:deposit}
Under the model assumptions above, the rule
\eqref{eq:optimal-retention}--\eqref{eq:allocation-rule} solves
Problem~\eqref{eq:design-lp} and attains the actionwise upper bound obtained by
dropping intended-user monotonicity and IR. Under this rule, the designated
profile is a pure-strategy Bayesian Nash equilibrium: every
intended-user type $\kappa>0$ uniquely prefers full commitment, while entrant
non-entry is a best response but is not uniquely selected.
\end{theorem}

\begin{proof}[Proof sketch]
Fix $a>0$, write $p=F(a)$, and relax intended-user monotonicity and IR at
this action, using the coefficients in \eqref{eq:phiR}--\eqref{eq:phiX}.
For $q\in(0,1)$, choose $\lambda=(K-1)/q$. Maximizing
$KU_C(a)-\lambda U_D(a)$ pointwise refunds below $F(m)=q$, retains in full
above it, and allocates whenever rank-feasible and
$F(m)<(v_C/v_D)q$. The entrant payoff under this maximizer is
\begin{equation}
 h_a(q)=v_D\min\!\left\{p,\frac{v_C}{v_D}q\right\}^{K}-a(1-q^K).
 \label{eq:entrant-root}
\end{equation}
It is strictly increasing, with limits $-a$ as $q\downarrow0$ and $v_Dp^K$
as $q\uparrow1$, so its unique zero is $q_R(a)$ in \eqref{eq:qR};
\eqref{eq:qbar} is the associated allocation cutoff. The candidate attains
the pointwise maximum and has $U_D(a)=0$.
Its intended-user expected payoff is
\begin{equation}
\begin{aligned}
 U_C^*(a)={}&v_C\min\!\left\{F(a),\frac{v_C}{v_D}q_R(a)\right\}^{K-1}\\[-0.2em]
          &-a\left[1-q_R(a)^{K-1}\right].
\end{aligned}
 \label{eq:optimal-utility}
\end{equation}
Since any rule satisfying the entrant constraint has $U_D(a)\leq0$,
$KU_C(a)\leq KU_C(a)-\lambda U_D(a)\leq KU_C^*(a)$. This is the actionwise
weak-duality certificate.

The branch formulas agree at every switch. For every
$\varepsilon\in(0,1)$, $U_C^*$ is absolutely continuous on
$[\varepsilon,1]$ with $U_C^{*\prime}(a)>0$ almost everywhere,
and $U_C^*(a)\to0$ as $a\downarrow0$. Thus $U_C^*$ is continuous, strictly
increasing, and positive for $a>0$, restoring monotonicity and IR; full algebra
is in the technical appendix. Integrating proves optimality in
Problem~\eqref{eq:design-lp}. The candidate has $U_D(a)=0$ at every action, so
Proposition~\ref{prop:target-equilibrium} gives the stated equilibrium and
best-response properties.
\end{proof}
The closed-form result requires no distributional regularity beyond the
maintained assumptions. Section~\ref{sec:joint-design} shows that the same
value also bounds joint rule-and-policy design.

\paragraph{Why optimality can precede necessity.}
The two-user uniform benchmark below makes the tradeoff transparent. In its
middle regime, the best leader-always-allocate rule supports non-entry by
retaining burdens over a broader set of states. Since the
entrant-to-intended-user likelihood ratio rises with congestion, those
additional lower-congestion states receive relatively more weight under
intended-user play. The optimum instead refunds them and concentrates
screening in the upper tail. It may sacrifice allocation in a small, highly
entrant-diagnostic region to save retained burdens over a much larger region.
Holdback can therefore improve surplus before it is necessary for feasibility;
Figure~\ref{fig:endogenous-reserve} later quantifies this tradeoff.

\paragraph{Holdback activation.}
For own action $a$, the holdback band is
$\mathcal R(a)=\{m:\ \bar q_R(a)\leq F(m)<F(a)\}$. Its conditional
probability under the intended-user state distribution is
$\rho^*(a)=\bigl[F(a)^{K-1}-\bar q_R(a)^{K-1}\bigr]_+$, where
$[z]_+\equiv\max\{z,0\}$. Because at most one account is strictly highest, the
ex ante probability that current allocation is withheld is
$\rho^*=K\int_0^1\rho^*(a)\dd F(a)$.

\begin{corollary}[Holdback activation from primitives]
\label{cor:activation-general}
For $a>0$, define
\begin{equation}
 \Psi_F(a)=F(a)^K\left[\frac{v_D}{a}
              +\left(\frac{v_D}{v_C}\right)^K\right].
 \label{eq:psi}
\end{equation}
The optimal rule uses a positive holdback band at action $a$ iff
$\Psi_F(a)>1$. Consequently, $\rho^*>0$ iff the set
$\{a:\Psi_F(a)>1\}$ has positive $F$-measure.
\end{corollary}
\begin{proof}
Holdback is active iff $\bar q_R(a)<F(a)$, equivalently
$q_R(a)<(v_D/v_C)F(a)$. Direct comparison of the two entries in
\eqref{eq:qR} shows that the second strictly exceeds the first exactly when it
lies below $\bigl[(v_D/v_C)F(a)\bigr]^K$. Thus the strict holdback inequality
selects $q_R(a)^K=a/(a+B_K)$; substituting and rearranging gives
$\Psi_F(a)>1$.
Integrating over own actions gives the ex ante statement.
\end{proof}

\paragraph{Necessity boundary.}
Corollary~\ref{cor:activation-general} identifies when the class-optimal rule
chooses holdback. This differs from asking when every pointwise-feasible rule
supporting non-entry must use it; the next result gives that separate boundary.

\begin{proposition}[When allocation holdback is necessary]
\label{prop:reserve-necessary}
Under the normalized retained-burden cap of one, if $v_D>1$, every
pointwise-feasible local rule that weakly supports non-entry has positive ex
ante allocation holdback under the designated profile. In particular, always
allocating to every unique leader cannot support non-entry.
\end{proposition}
\begin{proof}
For action $a$, let $A_D(a)=\int x(a,m)\dd F(m)^K$. The burden cap and
$U_D(a)\leq0$ imply $A_D(a)\leq a/v_D$. Entrant-state holdback relative to
always allocating is therefore at least $F(a)^K-a/v_D$, which is positive on
a neighborhood of one when $v_D>1$. For every action in that neighborhood,
the measures induced by $F^K$ and $F^{K-1}$ on rival states are mutually
absolutely continuous. Hence positive entrant-state holdback implies positive
intended-user-state holdback; integrating over own actions proves the claim.
\end{proof}
This necessity is conditional on the designated non-entry outcome; it does
not say that exclusion is desirable under an objective that also values
entrant access.

\begin{figure*}[!t]
\centering
\includegraphics[width=0.94\textwidth]{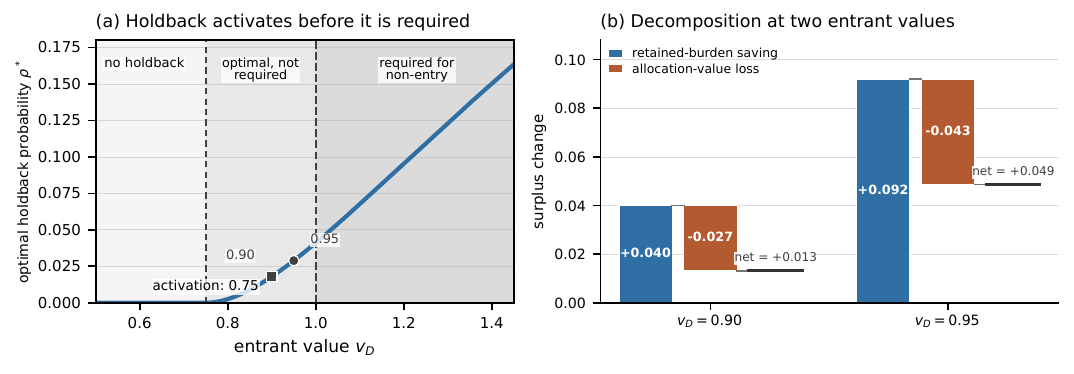}
\caption{Holdback activation and two-point surplus decomposition for $K=2$,
$F(a)=a$, and $v_C=1.5$. Panel (a) marks activation at $v_D=0.75$, the two
comparison values $v_D=0.90$ and $0.95$, and the $v_D=1.00$ boundary above
which a leader-always-allocate rule cannot support non-entry. Panel (b) decomposes
the class-optimal surplus gain relative to the best leader-always-allocate
comparator feasible for Problem~\eqref{eq:design-lp}. Each
waterfall rises with the retained-burden saving and falls with foregone
intended-user allocation value; its labeled right endpoint is the net gain.}
\label{fig:endogenous-reserve}
\end{figure*}

\begin{proposition}[Optimal before necessary in the two-user uniform benchmark]
\label{prop:three-regimes}
Let $K=2$ and $F(a)=a$, and let $\widehat v_D^R(v_C)\in(0,1)$ be the unique
solution to
\begin{equation}
 v+\left(\frac{v}{v_C}\right)^2=1.
 \label{eq:uniform-threshold}
\end{equation}
Then:
\begin{enumerate}
 \item If $v_D\leq\widehat v_D^R(v_C)$, the class-optimal rule uses no
 holdback.
 \item If $\widehat v_D^R(v_C)<v_D\leq1$, the class-optimal rule uses
 positive holdback and yields strictly higher intended-user surplus than any
 holdback-free feasible solution to Problem~\eqref{eq:design-lp}, although
 always allocating to the unique leader can still support entrant non-entry.
 \item If $1<v_D<v_C$, positive holdback is necessary to support non-entry.
\end{enumerate}
\end{proposition}
\begin{proof}
Under uniform capacities, $\Psi_F(a)=v_Da+(v_D/v_C)^2a^2$ is strictly
increasing, so Corollary~\ref{cor:activation-general} gives the threshold in
\eqref{eq:uniform-threshold}. The third statement is
Proposition~\ref{prop:reserve-necessary}. For $v_D\leq1$, always allocating
to the strict leader while retaining the full commitment in every state is
feasible: an intended-user action $a$ obtains $(v_C-1)a$ and the entrant
obtains $v_Da^2-a\leq0$, so holdback-free support exists throughout the
middle regime. No holdback-free rule feasible for
Problem~\eqref{eq:design-lp} can attain the class optimum there because
$\Phi_X$ is strictly negative on the positive-measure holdback band, where
allocating cannot attain the dual upper bound.
\end{proof}

For the illustrative parameterization $v_C=1.5$,
equation~\eqref{eq:uniform-threshold} gives the exact activation threshold
$\widehat v_D^R=0.75$.

\paragraph{Benchmark surplus decomposition.}
\label{sec:experiments}

Figure~\ref{fig:endogenous-reserve} first separates the point where holdback
becomes optimal from the boundary above which it becomes necessary. At each of
the two comparison values, the best leader-always-allocate comparator adds only
the requirement to allocate to every unique leader to the local rule class and
constraints of Problem~\eqref{eq:design-lp}. At $v_D=0.90$ and $0.95$,
holdback occurs in $1.793\%$ and $2.887\%$ of designated-profile states. In
both waterfalls, retained-burden savings exceed foregone intended-user
allocation value,
yielding net gains of $0.013$ and $0.049$ ($1.192\%$ and $4.685\%$).
These points illustrate the mechanism rather than calibrate representative
effect sizes; exact comparator construction and accounting are in the
technical appendix.

\section{Joint-Design Bound and Robustness}
\label{sec:joint-design}

Theorem~\ref{thm:deposit} solves the local-rule problem at the designated
profile. We next show that its value also bounds joint choice of a rule and
intended-user policy. A common, independently randomized policy is a
probability kernel $\sigma(\dd a\mid\kappa)$ supported on $[0,\kappa]$, with
private randomizations independent across intended users. A treatment vector
records every intended user's rule-assigned allocation probability and
retained burden.

\begin{lemma}[Full-commitment reduction]
\label{lem:full-reduction}
Fix a pointwise-feasible measurable local rule $(x,r)$ and a common,
independently randomized, capacity-feasible intended-user policy $\sigma$
against which entrant non-entry is a best response. Then there exists a
pointwise-feasible measurable local rule $(\widehat x,\widehat r)$ whose
full-commitment profile reproduces the same unconditional treatment-vector
distribution and aggregate expected intended-user surplus, while preserving
entrant non-entry as a best response.
\end{lemma}
\begin{proof}[Proof sketch]
Let $\Theta\sim F$, draw $A$ from $\sigma(\cdot\mid\Theta)$, and let $G$ be the
CDF of $A$. Capacity feasibility gives $G\geq F$. With $Q_G$ the generalized
quantile of $G$, set $\phi=Q_G\circ F$; then $\phi$ is nondecreasing,
$\phi(a)\leq a$, and $\phi(\Theta)\sim G$. Define
$(\widehat x,\widehat r)(a,m)=(x,r)(\phi(a),\phi(m))$.
The burden and rank constraints are preserved, while monotonicity gives
$\phi(\max_j\Theta_j)=\max_j\phi(\Theta_j)$. Hence full commitments reproduce
the i.i.d.\ action distribution, including atoms and ties, and therefore the
unconditional treatment vector and surplus. A transformed entrant action $a$
earns its original payoff from $\phi(a)$, preserving non-entry.
\end{proof}

\begin{corollary}[Joint-design bound]
\label{cor:joint-design}
The value attained in Theorem~\ref{thm:deposit} equals the supremum over all
pairs consisting of a pointwise-feasible measurable local rule and a common,
independently randomized, capacity-feasible intended-user policy for which
entrant non-entry is a best response.
\end{corollary}
\begin{proof}
Lemma~\ref{lem:full-reduction} maps every such pair to a pointwise-feasible
local rule under full commitments while preserving surplus and non-entry.
The relaxed actionwise certificate in the proof of
Theorem~\ref{thm:deposit} applies without intended-user monotonicity or IR and
bounds its surplus by the theorem's value, which the candidate rule paired
with full commitment attains.
\end{proof}

The enlarged comparison does not require the intended-user policy to be an
equilibrium. It therefore upper-bounds equilibrium-constrained joint design:
symmetric shading, interior pooling, nonmonotone policies, and independent
mixing cannot improve surplus. The technical appendix proves the stronger
almost-everywhere best-response preservation result when the original policy
is itself a symmetric equilibrium and explains how the reduction handles
pooling, atoms, and positive tie probability.

If only an upper bound $\bar v_D<v_C$ is known, designing at $\bar v_D$ solves
the robust joint-design problem requiring non-entry for every
$v_D\leq\bar v_D$; at lower values every positive commitment is strictly
worse than non-entry. The technical appendix gives the formal result.

\paragraph{Strict non-entry under account-creation costs.}
With an entrant-specific fixed account-creation cost $c>0$, every
participation action yields $-c$ while non-entry yields zero against the
designated profile. Thus non-entry is its unique best response there; this
does not imply uniqueness of the overall equilibrium.

\paragraph{False-name portfolios and trusted counts.}
Open identity raises boundary questions about finite identity splitting and
trusted account counts. The next corollary addresses the first for the
class-optimal rule; the technical appendix studies the second.
\begin{corollary}[Finite false-name portfolios]
\label{cor:finite-false-name}
At the designated full-commitment profile, suppose one economic entrant
controls finitely many accounts, values at most one allocation, and bears
additive retained burdens. Every such portfolio is weakly dominated by a
single account at its largest commitment. Hence entrant non-entry remains a
best response.
\end{corollary}
\begin{proof}[Proof sketch]
Under the cutoff rule, the top account's state payoff is weakly decreasing in
its strongest rival. Removing lower controlled identities therefore weakly
raises that payoff and removes only nonallocating accounts with nonnegative
burdens; resolving a top tie likewise helps. The identity $U_D(a)=0$ from
Theorem~\ref{thm:deposit} then bounds expected payoff by zero. The technical
appendix gives the full proof.
\end{proof}
This unilateral designated-profile result is not general
false-name-proofness.

\section{Conclusion}
\label{sec:conclusion}

Within the maintained local class, the optimal rule uses the strongest rival
as a congestion signal. Unlike a conventional reserve that responds to a low
leading bid, it may reject an otherwise eligible leader when the runner-up
commitment is unusually strong. This targeted upper-tail nonallocation is the
central design result. The leader-always-allocate comparison further shows
that it can improve intended-user surplus before it is necessary to support
entrant non-entry.
Together with the actionwise certificate, the full-commitment reduction
extends the same surplus bound to common, independently randomized,
capacity-feasible one-account policies under which entrant non-entry is a best
response. These results assume a common direct-use value, risk neutrality, and
additive enforceable burdens; non-entry is weak without entrant-specific
account-creation costs. When commitment is a poor proxy for intended-use value,
an access-oriented open-identity system may use the runner-up commitment as a
congestion signal and withhold current allocation from an otherwise eligible
unique leader, without fallback allocation.

\ifx\ArxivAcknowledgments\empty\else
\section*{Acknowledgments}
\ArxivAcknowledgments
\fi

\bibliographystyle{plainnat}
\bibliography{references}

\clearpage
\appendix
\setcounter{section}{0}
\setcounter{theorem}{0}
\setcounter{equation}{0}
\setcounter{figure}{0}
\setcounter{table}{0}
\renewcommand{\thesection}{S.\arabic{section}}
\renewcommand{\thesubsection}{\thesection.\arabic{subsection}}
\renewcommand{\thetheorem}{S.\arabic{theorem}}
\renewcommand{\theequation}{S.\arabic{equation}}
\renewcommand{\thefigure}{S.\arabic{figure}}
\renewcommand{\thetable}{S.\arabic{table}}
\renewcommand{\theHsection}{S.\arabic{section}}
\renewcommand{\theHsubsection}{S.\arabic{section}.\arabic{subsection}}
\renewcommand{\theHtheorem}{S.\arabic{theorem}}
\renewcommand{\theHequation}{S.\arabic{equation}}
\renewcommand{\theHfigure}{S.\arabic{figure}}
\renewcommand{\theHtable}{S.\arabic{table}}

\begin{center}
{\Large\bfseries Technical Appendix}
\end{center}
\addcontentsline{toc}{section}{Technical Appendix}

\section{Design Program and Incentive Reduction}

There are $K\geq2$ intended users, each commitment-limited, with the
commitment cap normalized to one, common direct-use value $v_C>1$, and
independent commitment capacities distributed according to a strictly increasing, absolutely
continuous CDF $F$ on $[0,1]$, with density $f>0$ almost everywhere. A
scalable entrant obtains value $v_D\in(0,v_C)$ from obtaining the
opportunity for operational or commercial use. These values, $K$, and $F$ are common
knowledge; an intended user's commitment capacity is private.
Subscript $C$ labels intended-user quantities and $D$ labels the scalable
entrant. Account ownership and controller role are not observable or verifiable, so the
operator cannot condition on whether a submitted account is controlled by an
intended user or by the entrant.

For a submitted profile $\boldsymbol a=(a_j)_{j\in N}$, define account $i$'s
strongest rival as
$M_i(\boldsymbol a_{-i})=\max_{j\in N\setminus\{i\}}a_j$. A measurable anonymous
local mechanism applies common functions $x$ and $r$, giving account $i$
allocation probability $x(a_i,M_i)\in[0,1]$ and retention $r(a_i,M_i)\in[0,a_i]$.
Allocation is feasible only when $M_i<a_i$. The rule is account-count-independent
and otherwise depends only on $(a_i,M_i)$. In the focal-account calculations,
$a$ denotes an own action and $m$ a realization of $M_i$. Intended users and
the entrant are risk neutral, burdens are linear and additive, and
state-dependent retention is enforceable.
We call a measurable local rule pointwise feasible if
$0\leq x(a,m)\leq\one\{m<a\}$ and $0\leq r(a,m)\leq a$ for every
$(a,m)\in[0,1]^2$.

Commitment capacity is conceptually distinct from use value and caps only
enforceable exposure. An intended user may therefore value access highly while
having little enforceable liquidity. Each intended user privately observes her
commitment capacity. Without observing other intended users' realized
capacities or any other submitted commitments, intended users and the entrant
choose commitments simultaneously; the entrant may abstain. The
committed operator then observes the profile and applies the rule. Entrant
deviations are therefore fixed actions $a$, not policies $a(m)$ chosen after
observing congestion. A rank-eligible state with $x(a,m)=0$ withholds the
current opportunity and provides no continuation value in the model.

Against the designated full-commitment intended-user profile, an intended user
faces the maximum of $K-1$ independent capacities, whose CDF is $F^{K-1}$,
whereas a fixed entrant action faces the maximum of all $K$ capacities, whose
CDF is $F^K$. The action utilities are
\begin{align}
 U_C(a)&=\int_0^1\left[v_Cx(a,m)\one\{m<a\}-r(a,m)\right]
       \dd F(m)^{K-1},\label{eq:supp-UC}\\
 U_D(a)&=\int_0^1\left[v_Dx(a,m)\one\{m<a\}-r(a,m)\right]
       \dd F(m)^K.\label{eq:supp-UD}
\end{align}
A type $\kappa$ can choose every $a\leq\kappa$. Hence the designated strategy
$a=\kappa$ is optimal for every type if and only if $U_C$ is nondecreasing. The
normalization $U_C(0)=0$ then makes nonnegative utility equivalent to Bayesian
interim individual rationality; ex post IR is not imposed. Weak support for
non-entry is the continuum of inequalities
$U_D(a)\leq0$. The design problem is the infinite-dimensional linear program
\[
 \max_{x,r}\ K\int_0^1U_C(a)\dd F(a)
 \quad\text{subject to}\quad
 U_C\ \text{nondecreasing},\quad U_C\geq0,\quad U_D\leq0,
\]
together with pointwise allocation feasibility and $0\leq r(a,m)\leq a$.
The proof below constructs a feasible candidate and an explicit weak-duality
certificate. It does not invoke strong duality in a function space.

\subsection{Full-commitment reduction and equilibrium preservation}

The next result distinguishes a reduction of the design problem from an
assumption about behavior under a fixed rule. A common, independently
randomized policy is a probability kernel $\sigma(\dd a\mid\kappa)$ supported
on $[0,\kappa]$; private randomizations are independent across intended users.
Pure strategies are degenerate kernels. When equilibrium is invoked for such
a policy, it means a Bayesian equilibrium in behavioral strategies: for every
type, the kernel assigns probability one to best responses in that type's
feasible action set. For any realization of intended-user actions when the
entrant stays out, the treatment vector records each intended user's
rule-assigned allocation probability and retained burden. Aggregate expected
intended-user surplus is the expected
sum of allocation value minus retained burden across the $K$ intended users.

\begin{proposition}[Full-commitment reduction]
\label{prop:supp-full-normalization}
Fix a pointwise-feasible measurable local rule $(x,r)$ and any common,
independently randomized capacity-feasible intended-user policy $\sigma$. Let
$G$ be its unconditional action CDF. If entrant non-entry is a best response
against the induced action distribution, there is another pointwise-feasible
measurable local rule $(\widehat x,\widehat r)$ with the following properties.
\begin{enumerate}
\item Full capacity commitments under $(\widehat x,\widehat r)$ induce the same
unconditional distribution of treatment vectors, and
the same aggregate expected intended-user surplus, as $\sigma$ under $(x,r)$.
\item Entrant non-entry remains a best response against the transformed
intended-user action distribution.
\item If the profile $(\sigma,\ldots,\sigma,\varnothing)$ is a symmetric
Bayesian equilibrium, full commitment under $(\widehat x,\widehat r)$ is a best
response for $F$-almost every intended-user type.
\end{enumerate}
The result permits nonmonotone pure strategies, mixing, pooling, and atoms in
$G$.
\end{proposition}

The first claim is unconditional. It does not preserve the joint distribution
of private capacities and treatments or type-by-type interim payoffs. The final
claim is only an almost-everywhere best-response statement; it need not preserve
an equilibrium under the typewise convention used in the main paper.

\begin{proof}
Let $\Theta\sim F$, draw $A$ from $\sigma(\cdot\mid\Theta)$, and write
\[
 G(t)=\Prb(A\leq t).
\]
Capacity feasibility gives $A\leq\Theta$ almost surely, so
$G(t)\geq F(t)$ for every $t$. Let
\[
 Q_G(u)=\inf\{t\in[0,1]:G(t)\geq u\},\qquad
 \phi(a)=Q_G(F(a)),
\]
with $Q_G(0)=0$. Since $F$ is continuous and strictly increasing,
$F(\Theta)$ is uniform. Standard quantile coupling therefore gives
\[
 \phi(\Theta)\sim G,\qquad
 \phi\ \text{nondecreasing},\qquad
 \phi(a)\leq Q_F(F(a))=a.
\]
The last inequality follows from $G\geq F$.

Define the transformed rule by
\begin{equation}
 \widehat x(a,m)=x(\phi(a),\phi(m)),\qquad
 \widehat r(a,m)=r(\phi(a),\phi(m)).
 \label{eq:supp-transformed-rule}
\end{equation}
It is measurable, anonymous, account-count-independent, and uses only the own
action and strongest rival. Moreover,
$0\leq\widehat r(a,m)\leq\phi(a)\leq a$. If $m\geq a$, monotonicity of
$\phi$ gives $\phi(m)\geq\phi(a)$, so original strict-leader feasibility gives
$\widehat x(a,m)=0$. Thus the transformation stays in the local rule class.

Under full capacity commitments, the effective commitments
$\phi(\Theta_1),\ldots,\phi(\Theta_K)$ are independent draws from $G$. Also,
\[
 \phi\!\left(\max_{j\ne i}\Theta_j\right)
   =\max_{j\ne i}\phi(\Theta_j).
\]
Consequently, applying \eqref{eq:supp-transformed-rule} to the full-commitment
profile gives exactly the same distribution of action-dependent allocations
and retentions as applying $(x,r)$ to the original policy. Because intended
users have a common direct-use value and capacity affects utility only through
the action constraint, expected intended-user surplus is identical.

For completeness, define the original fixed-action payoffs against $G$ by
\begin{align}
 u_G(c)&=\int_{[0,1]}[v_Cx(c,m)-r(c,m)]\,\dd G(m)^{K-1},
 \label{eq:supp-general-u}\\
 d_G(c)&=\int_{[0,1]}[v_Dx(c,m)-r(c,m)]\,\dd G(m)^K.
 \label{eq:supp-general-d}
\end{align}
These are Stieltjes integrals and therefore include atoms and tie
probabilities. A transformed entrant action $a$ yields
$d_G(\phi(a))$. Original non-entry support says $d_G(c)\leq0$ for every
$c\in[0,1]$, proving the second claim.

Suppose now that $\sigma$ is a symmetric Bayesian equilibrium. Under the joint
law of $(\Theta,A)$, a chosen action is a best response in $[0,\Theta]$ almost
surely. Since $A\leq\Theta$, this implies
\[
 A\in\mathcal L
 :=\left\{c:u_G(c)\geq u_G(b)\ \text{for every }b\leq c\right\}
 \quad\text{almost surely}.
\]
Because $u_G$ is Borel measurable, this lower-record set is universally
measurable; probability-one statements use the completion of the relevant
probability measure. Under the additional regularity that $u_G$ is
upper semicontinuous, $\mathcal L$ is Borel. No continuity is needed for the
surplus-preservation or upper-bound parts of the reduction.
Hence $\Prb(A\in\mathcal L)=1$, and, because $\phi(\Theta)\sim G$,
$\phi(\kappa)\in\mathcal L$ for $F$-almost every $\kappa$. Under the
transformed rule, full capacity commitment by
type $\kappa$ yields $u_G(\phi(\kappa))$. A deviation to $b\leq\kappa$ yields
$u_G(\phi(b))$; because $\phi$ is nondecreasing,
$\phi(b)\leq\phi(\kappa)$, and membership in $\mathcal L$ makes the deviation
unprofitable. This proves the final claim.
\end{proof}

Pooling causes no difficulty. A flat part of $\phi$ can map distinct submitted
commitments to the same effective commitment. The original strict-leader
constraint then assigns zero at the effective tie, so the transformed rule
reproduces the tie loss as deliberate holdback. Full commitment need only be a
weak best response after such a reduction. The optimal rule derived below
separately delivers a unique full-commitment best response through strict
increase of $U_C^*$.

\section{Congestion Sorting and Holdback Necessity}

\begin{lemma}[Upper-tail likelihood ratio]
An intended user choosing a fixed action faces strongest-rival density
\[
 g_C(m)=(K-1)F(m)^{K-2}f(m),
\]
whereas the entrant choosing the same fixed action faces
\[
 g_D(m)=KF(m)^{K-1}f(m).
\]
The likelihood ratio is $g_D(m)/g_C(m)=[K/(K-1)]F(m)$ almost everywhere and is
strictly increasing.
\end{lemma}
\begin{proof}
The maximum of $n$ independent capacities has CDF $F^n$ and density
$nF^{n-1}f$. Substitution gives the displayed densities and ratio. Since $F$
is strictly increasing, so is the ratio. For the rearrangement implication,
write the intended-user-state measure as $\mu_C$ and the entrant-state measure
as $\mu_D=L\mu_C$, where $L$ is increasing. If two equal-$\mu_C$ pieces of a
burden are exchanged from states $m_1<m_2$, the change in entrant screening is
proportional to $L(m_2)-L(m_1)\geq0$. Iterating such exchanges assigns any
fixed expected intended-user burden to an upper set. This is the
monotone-likelihood-ratio or Neyman--Pearson ordering used in the pointwise
program.
\end{proof}

For a fixed action, attach multiplier $\lambda\geq0$ to the entrant constraint. Up
to a common positive factor, the marginal state values of maximal retention
and current allocation are
\[
 \Phi_R(m)=\lambda F(m)-(K-1),
 \qquad
 \Phi_X(m)=(K-1)v_C-\lambda v_DF(m).
\]
The first rises and the second falls with congestion. Moreover,
$\Phi_X(m)<0$ implies
$\lambda F(m)>(K-1)v_C/v_D>K-1$, hence $\Phi_R(m)>0$. Thus allocation is never
withheld while the less destructive retention instrument remains unused.

\begin{proposition}[Holdback necessity under bounded retained burden]
If $v_D>1$, every pointwise-feasible local rule that weakly supports non-entry
and has retained burden bounded by one must use positive ex ante holdback
under the designated profile.
\end{proposition}
\begin{proof}
For action $a$, let $A_D(a)=\int x(a,m)\dd F(m)^K$. Since $r(a,m)\leq a$,
$0\geq U_D(a)\geq v_DA_D(a)-a$, whence $A_D(a)\leq a/v_D$. Entrant-state
holdback relative to always allocating to the unique leader is at least $F(a)^K-a/v_D$. This is
$1-1/v_D>0$ at $a=1$ and remains positive on a neighborhood of one. The
measures induced by $F^K$ and $F^{K-1}$ on rival states are mutually
absolutely continuous for every action in that neighborhood, so positive
entrant-state holdback implies positive intended-user-state holdback there.
Integrating over its positive $F$-measure proves positive ex ante holdback.
\end{proof}

\section{Complete Proof of the Optimal Local Rule}

Define
\[
 B_K=\frac{v_C^K}{v_D^{K-1}},
 \qquad p=F(a),
\]
and, for $a>0$,
\begin{align}
 q_R(a)^K&=\max\left\{1-\frac{v_D}{a}F(a)^K,
                         \frac{a}{a+B_K}\right\},\label{eq:supp-q}\\
 \bar q_R(a)&=\min\left\{1,\frac{v_C}{v_D}q_R(a)\right\}.
 \label{eq:supp-qbar}
\end{align}
The candidate retains the full amount iff $F(m)\geq q_R(a)$ and allocates iff
$F(m)<\min\{F(a),\bar q_R(a)\}$. Set $x(0,m)=r(0,m)=0$.

\subsection{Fixed-action weak-duality certificate}
Fix $a>0$ and abbreviate \eqref{eq:supp-UC}--\eqref{eq:supp-UD} by $U_C$ and
$U_D$. Since
\[
 \dd F(m)^{K-1}=(K-1)F(m)^{K-2}\dd F(m),
 \quad
 \dd F(m)^K=KF(m)^{K-1}\dd F(m),
\]
for every $\lambda\geq0$,
\begin{align}
 KU_C-\lambda U_D
 =\int_0^1 KF(m)^{K-2}\Big(&[(K-1)v_C-\lambda v_DF(m)]
       x(a,m)\one\{m<a\}\notag\\[-0.2em]
 &+[-(K-1)+\lambda F(m)]r(a,m)\Big)\dd F(m).
 \label{eq:supp-dual}
\end{align}
Choose $\lambda=(K-1)/q$ with $q\in(0,1)$. The coefficient on retention is
negative below $F(m)=q$ and positive above it, so pointwise maximization sets
$r=0$ below the cutoff and $r=a$ above it. The allocation coefficient is
positive below $F(m)=(v_C/v_D)q$ and negative above it. Combining this sign
with rank feasibility yields the three bands: refund and allocate, retain and
allocate, and retain and withhold current allocation.

Under that rule, entrant payoff from action $a$ is
\begin{equation}
 h_a(q)=v_D\min\left\{p,\frac{v_C}{v_D}q\right\}^{K}
        -a(1-q^K).
 \label{eq:supp-h}
\end{equation}
It is continuous and strictly increasing. On the branch where the minimum is
$p$, only the retained-burden term varies and has positive derivative. On
the other branch, both the value term and the release of retained burden
have positive derivatives. Also $h_a(0)=-a$ and $h_a(1)=v_Dp^K>0$, so there is
a unique zero.

If $(v_C/v_D)q\geq p$, the zero of \eqref{eq:supp-h} is
\[
 q^K=1-\frac{v_D}{a}p^K.
\]
If $(v_C/v_D)q<p$, it is
\[
 v_D\left(\frac{v_C}{v_D}\right)^Kq^K-a(1-q^K)=0,
 \qquad
 q^K=\frac{a}{a+B_K}.
\]
At the branch boundary $q=(v_D/v_C)p$, the two formulas coincide. Because the
zero is unique, the consistent solution is precisely the maximum in
\eqref{eq:supp-q}. Both entries of the maximum are at most one, the second is
strictly positive, and therefore $q_R(a)\in(0,1)$.

The candidate attains the pointwise maximum in \eqref{eq:supp-dual} and has
$U_D(a)=0$. For any rule satisfying the entrant constraint,
$U_D(a)\leq0$, so
\[
 KU_C(a)\leq KU_C(a)-\lambda U_D(a)
 \leq KU_C^*(a)-\lambda\cdot0=KU_C^*(a).
\]
This proves fixed-action optimality in the relaxation that omits the
designated-profile monotonicity and interim-IR constraints.
The candidate utility is allocation value minus retained burden:
\begin{equation}
 U_C^*(a)=v_C\min\left\{p,\frac{v_C}{v_D}q_R(a)\right\}^{K-1}
       -a\left[1-q_R(a)^{K-1}\right].
 \label{eq:supp-UCstar}
\end{equation}
Equivalently,
\begin{equation}
 U_C^*(a)=
 \begin{cases}
 v_Cp^{K-1}-a+aq^{K-1},
   & (v_C/v_D)q\geq p,\\[0.3em]
 (a+B_K)q^{K-1}-a=a(q^{-1}-1),
   & (v_C/v_D)q<p,
 \end{cases}
 \label{eq:supp-UCbranches}
\end{equation}
where $q=q_R(a)$.

\subsection{Designated-profile incentive compatibility and individual rationality}
Define
\[
 H_K(q)=\frac{q^K+(K-1)-Kq}{Kq}.
\]
For every $a>0$, both terms inside the maximum defining $q_R(a)^K$ are
strictly below one, while the second is strictly positive. Hence
$q_R(a)\in(0,1)$. The numerator defining $H_K$ equals zero at $q=1$ and has
derivative $K(q^{K-1}-1)<0$ on $(0,1)$, so $H_K(q_R(a))>0$.

On a holdback-free interval, $q^K=1-(v_D/a)p^K$. At points of
differentiability,
\[
 Kq^{K-1}q'
 =-v_D\left(\frac{Kp^{K-1}f(a)}{a}-\frac{p^K}{a^2}\right).
\]
Differentiating the first line of \eqref{eq:supp-UCbranches}, substituting this
identity, and using $v_Dp^K/a=1-q^K$ yields
\begin{equation}
 \frac{\dd U_C^*(a)}{\dd a}
 =(K-1)p^{K-2}f(a)\frac{v_Cq-v_Dp}{q}+H_K(q).
 \label{eq:supp-UCprime-free}
\end{equation}
The branch condition $(v_C/v_D)q\geq p$ makes the first term nonnegative, so
the derivative is strictly positive almost everywhere.

On a holdback interval, $q^K=a/(a+B_K)$ and
$U_C^*(a)=a(q^{-1}-1)$. Log differentiation gives
$q'/q=B_K/[Ka(a+B_K)]$. Therefore
\[
 \frac{\dd U_C^*(a)}{\dd a}
 =q^{-1}-1-\frac{B_K}{Kq(a+B_K)}
 =H_K(q)>0.
\]
At a regime switch, $q=(v_D/v_C)p$ and the two lines in
\eqref{eq:supp-UCbranches} coincide.

For every $\varepsilon\in(0,1)$, the two expressions inside the maximum in
\eqref{eq:supp-q} are absolutely continuous on $[\varepsilon,1]$, and the
second is bounded away
from zero. Hence $q_R$ and $U_C^*$ are absolutely continuous there. The displayed
derivatives exist and are strictly positive almost everywhere, implying strict
increase on each such interval. Moreover, in the holdback-free branch
$0<U_C^*(a)\leq v_CF(a)^{K-1}\to0$, while in the holdback branch
$U_C^*(a)=a^{1-1/K}(a+B_K)^{1/K}-a\to0$. Together with
$U_C^*(0)=0<U_C^*(a)$ for $a>0$, this proves that $U_C^*$ is strictly increasing on
$[0,1]$.

In the holdback branch, $U_C^*(a)=a(q^{-1}-1)>0$. In the holdback-free branch,
\[
 a(1-q^{K-1})\leq a(1-q^K)=v_Dp^K<v_Cp^{K-1},
\]
so \eqref{eq:supp-UCstar} is strictly positive for every $a>0$. Thus, against
the designated opponents' play, full commitment is the unique best response of
each type $\kappa>0$; interim IR also holds. Since the candidate attains the fixed-action
upper bound for every $a$, integrating against $K\dd F(a)$ proves optimality
for intended-user surplus. The proof requires neither a hazard-rate restriction
nor ironing.

\begin{corollary}[Optimality over symmetric capacity-feasible policies]
\label{cor:supp-policy-optimality}
Enlarge the design comparison to all pairs consisting of a feasible local rule
and a common, independently randomized intended-user policy supported on each
type's feasible action set. Impose only that entrant non-entry be a best
response; the intended-user policy need not itself be an equilibrium. The
supremum of intended-user surplus in this enlarged comparison is
\[
 K\int_0^1 U_C^*(a)\,\dd F(a),
\]
and is attained by the optimal rule with full commitment.
\end{corollary}

\begin{proof}
Apply Proposition~\ref{prop:supp-full-normalization} to any pair in the
enlarged comparison. Under the transformed rule, let
$\widehat U_C(a)=u_G(\phi(a))$ denote intended-user action utility against full
capacity commitments. The transformation preserves expected surplus, so the
original pair yields $K\int_0^1\widehat U_C(a)\dd F(a)$. It also preserves
pointwise feasibility and gives transformed entrant payoff
$d_G(\phi(a))\leq0$ for every $a$. The fixed-action certificate
\eqref{eq:supp-dual} therefore implies
$\widehat U_C(a)\leq U_C^*(a)$ pointwise. This argument does not require
$\widehat U_C$ to be nondecreasing or nonnegative. Integration gives the upper
bound, and the rule characterized above attains it while making full
commitment the unique best response of every positive type.
\end{proof}

\section{Holdback Activation and the Three Regimes}

For own action $a$, current allocation is withheld exactly on
\[
 \mathcal R(a)=\{m:\bar q_R(a)\leq F(m)<F(a)\}.
\]
Under the intended-user state distribution, this set has conditional
probability
\begin{equation}
\rho^*(a)=\left[F(a)^{K-1}-\bar q_R(a)^{K-1}\right]_+.
\label{eq:supp-rho-a}
\end{equation}
Here $[z]_+\equiv\max\{z,0\}$.
Since unique-leader events are mutually exclusive across users, the ex ante
holdback probability is
\begin{equation}
 \rho^*=K\int_0^1\rho^*(a)\dd F(a).
 \label{eq:supp-rho}
\end{equation}

\begin{proposition}[General activation condition]
For $a>0$, define
\[
 \Psi_F(a)=F(a)^K\left[\frac{v_D}{a}
              +\left(\frac{v_D}{v_C}\right)^K\right].
\]
Then $\mathcal R(a)$ has positive probability iff $\Psi_F(a)>1$.
Consequently, $\rho^*>0$ iff the set $\{a:\Psi_F(a)>1\}$ has positive
$F$-measure.
\end{proposition}
\begin{proof}
The holdback band is nonempty iff $\bar q_R(a)<F(a)$. Because $F(a)\leq1$,
this is equivalent to $q_R(a)<(v_D/v_C)F(a)$. Direct comparison of the two
entries in \eqref{eq:supp-q} shows that the second strictly exceeds the first
exactly when it lies below $\bigl[(v_D/v_C)F(a)\bigr]^K$. Hence the strict
holdback inequality selects the second branch. Substituting
$q_R(a)^K=a/(a+B_K)$ gives
\begin{align*}
 \frac{a}{a+B_K}
 &<\left(\frac{v_D}{v_C}\right)^K F(a)^K\\
 \Longleftrightarrow\quad
 a
 &<a\left(\frac{v_D}{v_C}\right)^KF(a)^K+v_DF(a)^K\\
 \Longleftrightarrow\quad
 1
 &<F(a)^K\left[\frac{v_D}{a}
      +\left(\frac{v_D}{v_C}\right)^K\right].
\end{align*}
The ex ante statement follows from \eqref{eq:supp-rho}.
\end{proof}

\begin{corollary}[Uniform capacities]
If $F(a)=a$, let $\widehat v^R_{D,K}(v_C)$ be the unique root in $(0,1)$ of
\[
 v+\left(\frac{v}{v_C}\right)^K=1.
\]
Then the optimal holdback probability is zero for
$v_D\leq\widehat v^R_{D,K}(v_C)$ and positive above it.
\end{corollary}
\begin{proof}
Now
\[
 \Psi_F(a)=v_Da^{K-1}+\left(\frac{v_D}{v_C}\right)^Ka^K,
\]
which is strictly increasing in $a$. A positive-measure holdback set exists iff
$\Psi_F(1)>1$. The threshold equation has a unique root because its left side
is strictly increasing, equals zero at $v=0$, and exceeds one at $v=1$.
\end{proof}

For $K=2$ and $F(a)=a$, the threshold divides three economically distinct
regions. If $v_D$ is below the threshold, the class-optimal rule itself always
allocates to the unique leader. If the threshold is crossed while $v_D\leq1$, the optimum
uses positive holdback even though a holdback-free implementation exists. To
see feasibility, allocate whenever the focal action is strictly highest and retain the full
commitment in every state. An intended-user action $a$ then obtains
$(v_C-1)a$, which is increasing and nonnegative, while the entrant obtains
$v_Da^2-a\leq0$. Strict optimality of positive holdback follows from the
fixed-action dual certificate: on the holdback band, the allocation coefficient
is strictly negative, so a rule that allocates there cannot attain the dual
upper bound. Finally, if $v_D>1$, always allocating to the unique leader cannot support non-entry
because retained burden is capped at one.

At $v_C=3/2$, the two-user threshold solves
$v+(2v/3)^2=1$ and equals exactly $3/4$. For $K=3,5,10$, the corresponding
thresholds are $0.83036067$, $0.91536964$, and $0.98507900$.

\section{Holdback Probability and One-Dimensional Evaluation}

Let $u=F(a)$, $a(u)=F^{-1}(u)$, $q(u)=q_R(a(u))$, and
$b(u)=\min\{u,\bar q_R(a(u))\}$. Conditional on own quantile $u$, a constrained
account is allocated exactly when the maximum of the other $K-1$ independent
uniform quantiles is below $b(u)$. Its allocation probability is therefore
$b(u)^{K-1}$. It is retained exactly when that maximum is at least $q(u)$, so
its retention probability is $1-q(u)^{K-1}$. Summing symmetrically gives
\begin{align*}
 P_{\mathrm{alloc}}&=K\int_0^1b(u)^{K-1}\dd u,\\
 E_{\mathrm{ret}}&=K\int_0^1a(u)[1-q(u)^{K-1}]\dd u,\\
 \PS&=v_CP_{\mathrm{alloc}}-E_{\mathrm{ret}},\\
 \rho^*&=1-P_{\mathrm{alloc}}.
\end{align*}
Thus intended-user surplus, retained burden, and holdback are evaluated by
one-dimensional integration for every population size. The reported
high-precision numerical values use adaptive quadrature with absolute and
relative tolerances $2\times10^{-11}$.

\section{Best Leader-Always-Allocate Comparator}

The main within-class comparison holds fixed the local information, IC, IR,
retained-burden cap, and support of entrant non-entry, imposing only the
additional restriction that every unique leader be allocated. It is therefore
the restricted optimum obtained by prohibiting holdback, not an arbitrary
baseline. For $K=2$, $F(a)=a$, and $v_D\leq1$, impose
\[
 x(a,m)=\one\{m<a\}.
\]
At a fixed action $a$, the entrant obtains expected allocation value $v_Da^2$.
Upper-tail retention minimizes the intended user's expected burden subject to
offsetting that amount. The cutoff $q_0(a)$ therefore satisfies
\[
 a[1-q_0(a)^2]=v_Da^2,
 \qquad q_0(a)=\sqrt{1-v_Da},
\]
and the largest fixed-action utility consistent with the entrant constraint is
\begin{equation}
 g(a)=(v_C-1)a+a\sqrt{1-v_Da}.
 \label{eq:certain-ceiling}
\end{equation}

Global IC requires a nondecreasing action-utility schedule. Consequently, the
largest feasible schedule below the pointwise ceiling is its greatest
nondecreasing minorant,
\begin{equation}
 \underline g(a)=\inf_{t\in[a,1]}g(t).
 \label{eq:monotone-minorant}
\end{equation}
This schedule is attainable. At each action $a$, the upper-tail construction
attains $g(a)$, while retention in every state gives utility $(v_C-1)a$.
Mixing the two retention schedules therefore attains every utility in
$[(v_C-1)a,g(a)]$ without weakening the entrant constraint. Moreover, for
every $t\geq a$,
\[
 g(t)\geq(v_C-1)t\geq(v_C-1)a,
\]
and the infimum in \eqref{eq:monotone-minorant} includes $t=a$. Hence
$(v_C-1)a\leq\underline g(a)\leq g(a)$, so the required minorant lies in the
attainable interval. For the illustrative values below,
$g$ is concave, rises past $g(1)$, and then returns to it. Thus
$\underline g(a)=g(a)$ up to the first root $a_0$ of $g(a_0)=g(1)$ and equals
$g(1)$ thereafter.

At $(v_C,v_D)=(1.5,0.95)$,
\[
 a_0=0.644567078007,
 \qquad
 \PS_{\mathrm{LA}}
 =2\left[\int_0^{a_0}g(a)\dd a+(1-a_0)g(1)\right]
 =1.039422830098.
\]
Because this comparator always allocates, its expected retained burden is
$v_C-\PS_{\mathrm{LA}}=0.460577169902014$; the same $v_C$ is the burden-free
first-best intended-user-surplus value.
The class optimum has $\PS^*=1.088117457369$, so allowing holdback raises
intended-user surplus by $4.684775614\%$ and closes $10.572523011\%$ of the
remaining first-best intended-user-surplus gap. A grid implementation replaces
\eqref{eq:monotone-minorant} by suffix minima. At 800 action nodes it gives
$1.039422381$, within $4.50\times10^{-7}$ of the exact value.

For any rule $R$ at designated non-entry, let $A_R$ denote the intended-user
allocation probability and $B_R$ the aggregate expected retained burden. Then
\[
 \PS_R=v_CA_R-B_R.
\]
Because $A_{\mathrm{LA}}=1$ and $A^*=1-\rho^*$, the holdback-specific gain
has the implementation-invariant decomposition
\[
 \PS^*-\PS_{\mathrm{LA}}
 =
 (B_{\mathrm{LA}}-B^*)-v_C\rho^*.
\]
For the two values displayed in Panel (b), the components are
\begin{center}
\scriptsize
\begin{tabular}{rrrr}
\toprule
$v_D$ & $B_{\mathrm{LA}}-B^*$ & $v_C\rho^*$ &
$\PS^*-\PS_{\mathrm{LA}}$\\
\midrule
$0.90$ & $0.040020202$ & $0.026888464$ & $0.013131738$\\
$0.95$ & $0.092006583$ & $0.043311955$ & $0.048694627$\\
\bottomrule
\end{tabular}
\end{center}
The corresponding gains relative to the leader-always-allocate comparator are
$1.191762858\%$ and $4.684775614\%$. Panel (b) plots the two signed
contributions as side-by-side waterfalls: each right endpoint is the net
effect reported in the caption. A finer winner-versus-loser decomposition is
not invariant
because multiple retention schedules can implement the same utility minorant.

This comparator weakly supports non-entry. Separately, because $v_D<1$, the
simple leader-always-allocate rule that retains the full commitment in every state
gives entrant payoff $v_Da^2-a<0$ for every $a>0$. Hence positive holdback is
not required for the designated non-entry outcome at either displayed
illustrative point.

\section{Extensions and Implementation Scope}

\begin{proposition}[Design for an upper entrant-value bound]
Suppose the single entrant's value is unknown to the designer but belongs to
$[0,\bar v_D]$, where $0<\bar v_D<v_C$. The rule constructed at $\bar v_D$
maximizes intended-user surplus in the enlarged joint-design comparison
subject to non-entry being a best response for every value in this interval.
For every value strictly below $\bar v_D$, every positive commitment is
strictly worse than non-entry.
\end{proposition}
\begin{proof}
For a fixed rule and action, write $A(a)$ for allocation probability and
$R(a)$ for expected retained burden. Entrant payoff $vA(a)-R(a)$ is affine and
weakly increasing in $v$, so the upper-bound value gives the tightest
actionwise constraints. Requiring non-entry for every value in the interval is
therefore equivalent to imposing the constraints only at $\bar v_D$. The
intended-user-surplus objective at the designated non-entry profile does not
otherwise depend on entrant value, so the main optimality argument applies to
this robust comparison. For the constructed rule,
$U_{\bar v_D}(a)=0$ at every action and $A(a)>0$ for $a>0$. Hence
$v<\bar v_D$ gives $U_v(a)=-(\bar v_D-v)A(a)<0$ for every positive action.
At $a=0$, entry remains outcome-equivalent to non-entry.
\end{proof}

\begin{corollary}[Finite false-name portfolios]
Fix the designated intended-user profile. Suppose one economic entrant controls any
finite number of identities, values at most one delivered object, and sums all
burdens borne by its identities. Under the class-optimal rule, every such portfolio
is weakly dominated by submitting a single account at its largest commitment.
\end{corollary}
\begin{proof}
Fix a realization $M$ of the highest intended-user commitment and a finite
controlled portfolio. Let $a$ be its highest controlled commitment and let $b_{(2)}$ be
its second-highest, with $b_{(2)}=0$ for a singleton portfolio. The highest
controlled account faces strongest rival $\max\{M,b_{(2)}\}$, whereas a single
account with commitment $a$ faces $M$.

Under the cutoff rule following \eqref{eq:supp-qbar}, the highest account's
state payoff is weakly decreasing in opposing congestion. Depending on the
order of the allocation and retention cutoffs, its treatment moves from
refund-and-allocate either through retain-and-allocate or through
refund-and-withhold, and then to retain-and-withhold. Removing lower controlled
commitments therefore weakly raises the highest account's payoff. Every
removed identity obtains no additional single-unit allocation value and bears
a nonnegative retained burden, so its own contribution to portfolio payoff is
nonpositive.

If several controlled identities tie at the highest commitment while no
intended-user account ties them, the unique-leader convention gives them no
current-round allocation. Deleting all but one weakly increases the controller's
allocation opportunity and removes nonnegative burdens. A tie with the highest
intended-user commitment is a null event under the continuous capacity distribution. The comparison therefore holds almost surely; integrating over
$M$ and applying the single-account identity $U_D(a)=0$ from
the fixed-action certificate shows that every finite portfolio yields expected
payoff at most zero. Since non-entry yields zero, it remains a best response.
\end{proof}

The corollary is not dominant-strategy false-name-proofness for arbitrary
opponent profiles. It relies on unilateral splitting by one economic entrant,
the designated intended-user profile, single-unit demand, finitely many
identities, additive account burdens, and the maintained continuous capacity
distribution.

\subsection{Market-size uncertainty and count-based detection}

The count-dependent first-best construction uses the known number of intended
users as an off-path entry signal.

\paragraph{Known-market-size construction.}
When exactly $K$ commitments are positive, refund every commitment and
allocate to the unique leader. For any other positive-account count, allocate
nothing and retain every positive commitment in full. At the designated
profile, exactly $K$ commitments are positive almost surely, so the rule
attains first-best intended-user surplus with no retained burden or holdback.
An intended-user action $a$ obtains utility $v_CF(a)^{K-1}$, which is strictly
increasing. Any entrant action $a>0$ instead creates $K+1$ positive
commitments almost surely and earns $-a<0$; a zero commitment remains
payoff-equivalent to non-entry. This proves the count-dependent first-best
implementation.

The following result formalizes a market-size-uncertainty boundary for that
construction.

\begin{proposition}[Failure of first-best count detection under uncertainty]
Let the number of intended users be
$\widetilde K\in\{K,K+1\}$, independent of their capacities and not directly
observed by the mechanism, with
$q=\Prb(\widetilde K=K)\in(0,1)$. The rule observes all submitted commitments.
Suppose $F$ has no atom at zero and $qv_D>1-q$. No anonymous full-profile rule
that allocates at most one unit, allocates only to an account whose commitment
is strictly highest,
and retains at most each account's own commitment can both attain first-best
intended-user surplus in expectation over $\widetilde K$ at the designated
profile and weakly support entrant non-entry against every fixed commitment.
\end{proposition}
\begin{proof}
For either realized market size, intended-user surplus is bounded above by
$v_C$: at most one intended user receives the unit, and retained burdens are
nonnegative. Because both market sizes have positive probability, attaining
this upper bound in expectation requires full refunds and allocation to the
unique leader almost surely in each state. In particular, the rule gives that
generous treatment for $F^{\otimes(K+1)}$-almost every intended-user profile with
$K+1$ positive commitments.

Anonymity makes treatment depend on the submitted profile rather than on the
economic identity behind an account. Fubini's theorem therefore implies that,
for $F$-almost every fixed $a$, the same generous treatment applies
$F^{\otimes K}$-almost surely to the profile consisting of $K$ intended-user
commitments and the additional commitment $a$. Any neighborhood of one has
positive $F$-measure, so such admissible values of $a$ can be chosen
arbitrarily close to one.

When $\widetilde K=K$, an entrant using such an $a$ is refunded and wins
whenever all intended-user commitments are below $a$, an event of probability
$F(a)^K$. When $\widetilde K=K+1$, bounded retained burden gives the entrant
payoff at least $-a$, regardless of the allocation decision. Its expected
deviation payoff is consequently bounded below by
\[
 qv_DF(a)^K-(1-q)a.
\]
As $a$ approaches one through the admissible set, this lower bound converges to
$qv_D-(1-q)>0$. A sufficiently high admissible commitment is therefore a
strictly profitable deviation, contradicting weak non-entry support.
\end{proof}

This is a one-sided impossibility result. It does not imply that account-count
information is generally useless, nor does the reverse probability inequality
alone guarantee first-best implementation for an arbitrary capacity
distribution.

\subsection{Strategic roles and implementation boundaries}
The intended users' and entrant's incentive conclusions are asymmetric.
Against designated opponents' play, strict increase of $U_C^*$ makes full commitment
the unique best response of every intended-user type $\kappa>0$. By contrast,
the class-optimal rule has $U_D(a)=0$ for every $a>0$, while abstention also
yields zero. It therefore supports a non-entry equilibrium but does not select
non-entry from the entrant's best-response set. This observation alone neither
constructs nor rules out other equilibria because intended-user responses
may change when the entrant participates.

The designer maximizes $K\int U_C(a)\dd F(a)$ at the designated profile. Entrant
allocation value, operator revenue, and transfer-neutral credit for retained
payments are excluded. The operator is not a strategic player and can commit
to the state-contingent rule. The analysis also excludes an entrant that waits
to observe realized congestion and then chooses $a(m)$; the fixed-action
constraints $U_D(a)\leq0$ would not control such reactive entry. Coordinated
entrants, infinite identities, multi-unit or nonadditive values, nonadditive
burdens, and identity creation that expands an intended user's aggregate
commitment capacity remain outside the result.

\section{Finite-LP Cross-Check}
Let $u_i=i/M$, $i=0,\ldots,M$, and let $a_i=F^{-1}(u_i)$. Congestion is divided
into $M$ quantile cells. For cell $j=1,\ldots,M$, define
\[
 \Delta^C_j=\left(\frac{j}{M}\right)^{K-1}
             -\left(\frac{j-1}{M}\right)^{K-1},
 \qquad
 \Delta^D_j=\left(\frac{j}{M}\right)^K
             -\left(\frac{j-1}{M}\right)^K.
\]
The variables $x_{ij}$ and $r_{ij}$ are current allocation and retained
burden for action $a_i$ in state cell $j$. Feasibility imposes
$0\le r_{ij}\le a_i$ and $x_{ij}=0$ whenever the state cell is not strictly
below $u_i$. With trapezoid weight $w_i$ on the own-type grid, the objective is
\[
 \max\ K\sum_i w_i\sum_j \Delta^C_j
                  \bigl(v_Cx_{ij}-r_{ij}\bigr).
\]
For each action $i$, the entrant constraint is
\[
 \sum_j\Delta^D_j\bigl(v_Dx_{ij}-r_{ij}\bigr)\le0.
\]
Adjacent utility monotonicity implements grid IC, and nonnegative utility at
each node imposes IR. The LP has exactly $2M(M+1)$ treatment variables and $3M+2$
structural inequalities. It is used only as a finite-grid consistency check of
the analytic theorem and does not restrict the solution to a
prespecified cutoff class.

\begin{table}[t]
\centering
\caption{Local state-contingent LP convergence for $K=2$, $F(a)=a$,
$(v_C,v_D)=(1.5,0.95)$. High-precision quadrature gives
$\PS^*=1.088117457$.}
\small
\begin{tabular}{rrrr}
\toprule
$M$ & Variables & LP $\PS$ & Absolute gap\\
\midrule
20  & 840    & 1.087086 & $1.03\times10^{-3}$\\
40  & 3,280  & 1.087871 & $2.47\times10^{-4}$\\
80  & 12,960 & 1.088057 & $6.00\times10^{-5}$\\
120 & 29,040 & 1.088090 & $2.72\times10^{-5}$\\
\bottomrule
\end{tabular}
\label{tab:unrestricted-lp}
\end{table}

The objective gaps shrink monotonically, while entrant-payoff residuals remain
at solver tolerance. The $M=80$ LP already reproduces
the analytic objective to four decimal places. This is consistent with the
continuous theorem and shows that the discretized result is not caused by
imposing the analytic threshold form.

\section{Selected Numerical Comparisons}

\subsection{High-precision evaluation at the illustrative point}
For $K=2$, uniform capacities, and $(v_C,v_D)=(1.5,0.95)$, one-dimensional
quadrature gives
\[
 \PS^*=1.088117457,\quad
 \rho^*=0.028874637,\quad
 E_{\mathrm{ret}}^*=0.368570587.
\]
The ratio of intended-user surplus to the burden-free first-best value,
$\PS^*/v_C$, is $0.725411638$. These high-precision evaluations of the
analytically characterized rule are the source of the rounded statistics at
the illustrative point in the main paper.

\subsection{Capacity-capped second-price accommodation benchmark}
\label{sec:capped-second-price}

The main design program asks a mechanism to support entrant non-entry. A
different institutional response is to accommodate every participant through
a standard auction. To compare the two responses within the same commitment
technology, interpret each commitment capacity as a hard upper bound on a
payable bid. For an account with bid $a$ and highest opposing bid $m$, define
\[
 x^{\mathrm{SP}}(a,m)=\one\{a>m\},
 \qquad
 r^{\mathrm{SP}}(a,m)=m\one\{a>m\}.
\]
Thus a unique highest bidder receives the unit and pays the highest opposing
bid; every losing bid is refunded. A top tie receives no allocation. Because
the capacity distribution is continuous, this convention has no effect at the
equilibrium profile below.

\begin{proposition}[Capacity-capped second-price accommodation]
\label{prop:capped-second-price}
Let a controller have single-unit value $v>0$. Suppose no identity it controls
can bid above a common cap $c\in[0,1]$, creating identities does not relax that
ceiling, and all payments by its identities are charged additively to the
controller. Under the capacity-capped second-price rule, the single bid
\[
 b^*(v,c)=\min\{v,c\}
\]
weakly dominates every feasible single bid and every finite portfolio of
controlled bids. Consequently, the profile
\[
 b_i^*=\kappa_i
 \quad\text{for every intended user},\qquad
 b_D^*=d:=\min\{v_D,1\}
\]
is a weakly dominant-strategy equilibrium and is robust to finite false-name
portfolios under these common-cap, additive-payment, and single-unit-demand
assumptions. The equilibrium need not be unique. Let $\kappa_{(K)}$ and
$\kappa_{(K-1)}$ denote the largest and second-largest intended-user
capacities. The entrant's allocation probability, aggregate intended-user surplus net of the
winner's payment, and modeled expected allocated value are, respectively,
\[
 \Prb(D\text{ wins})=F(d)^K,
\]
\[
 S_C^{\mathrm{SP}}
 =
 \E\!\left[
   \bigl(v_C-\max\{d,\kappa_{(K-1)}\}\bigr)
   \one\{\kappa_{(K)}>d\}
 \right],
\]
and
\[
 V_{\mathrm{SP}}
 =v_C\!\left[1-F(d)^K\right]+v_DF(d)^K
 =v_C-(v_C-v_D)F(d)^K .
\]
\end{proposition}

\begin{proof}
Fix the highest bid $m$ submitted outside a controller's portfolio. The single
bid $b^*=\min\{v,c\}$ wins whenever $m<b^*$, pays $m$, and earns $v-m>0$. No
controlled portfolio can earn more in such a state: if one controlled account
wins, its payment is at least $m$, and single-unit demand supplies value only
once. If $m\geq b^*$ and $v\leq c$, any controlled bid that wins must induce a
payment of at least $m\geq v$ and hence earns at most zero. If $v>c$, no
feasible controlled bid can exceed $m$. Controlled ties also earn zero under
the tie convention. The single bid $b^*$ therefore weakly dominates every
finite portfolio state by state.

For an intended user, $v_C>1\geq\kappa_i$, so the capacity-constrained value bid is
$\kappa_i$. For the entrant it is $d=\min\{v_D,1\}$. Continuity makes the
highest bid unique almost surely. The entrant wins exactly when all intended
capacities lie below $d$, giving $F(d)^K$. If an intended user wins, the
payment is the larger of the entrant's bid and the second-largest intended
capacity, giving $S_C^{\mathrm{SP}}$. Weighting the two possible winner values
gives $V_{\mathrm{SP}}$.
\end{proof}

The entrant's expected utility at this profile is
\[
 S_D^{\mathrm{SP}}
 =
 \E\!\left[
   \bigl(v_D-\kappa_{(K)}\bigr)
   \one\{\kappa_{(K)}<d\}
 \right]>0 .
\]
Thus zero holdback here reflects accommodation, not successful screening.

At the same two-user uniform illustrative point, direct integration gives
\[
 S_C^{\mathrm{SP}}(d)
 =
 v_C(1-d^2)-\frac{1}{3}-d^2+\frac{4d^3}{3},
 \qquad 0\leq d\leq1.
\]
At $(v_C,v_D)=(1.5,0.95)$,
\[
 \Prb(D\text{ wins})=0.9025,\qquad
 V_{\mathrm{SP}}=1.003625,\qquad
 S_C^{\mathrm{SP}}=\frac{643}{12000}=0.053583\ldots .
\]
The entrant receives the unit most of the time because intended bids are
capacity-capped, even though $v_D<v_C$. At the designated non-entry outcome,
ignoring every retained monetary burden gives the following modeled expected
allocated value for the class-optimal
screening allocation:
\[
 v_CP_{\mathrm{alloc}}^*
 =\PS^*+E_{\mathrm{ret}}^*
 =1.456688 .
\]
This is the like-for-like comparison plotted in
Figure~\ref{fig:supp-accommodation}. For orientation,
the screening rule's distinct primary intended-user-surplus value, which
deducts retained burdens and excludes entrant value, is $1.088117$ and still
exceeds $V_{\mathrm{SP}}$. That cross-objective comparison deliberately gives
accommodation the more favorable accounting.

\begin{figure}[t]
\centering
\includegraphics[width=0.92\textwidth]{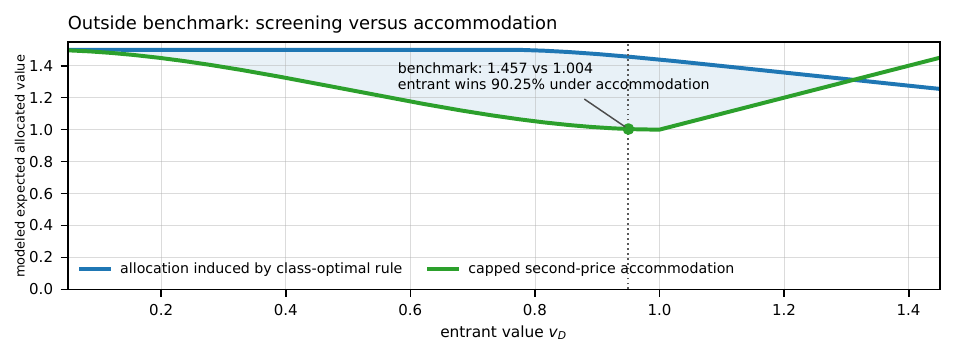}
\caption{Outside benchmark comparing modeled expected allocated value under
the allocation induced by the class-optimal screening rule at designated
non-entry with capacity-capped second-price accommodation. This comparison
changes both target behavior and the accounting convention.}
\label{fig:supp-accommodation}
\end{figure}

Four qualifications are essential. First, modeled expected allocated value
differs from the primary objective: it counts the entrant's value and ignores retained
monetary burdens. It is a model-internal accounting benchmark, not total
welfare. Second, the exclusion value
uses the designated non-entry outcome, whereas the auction value uses the
weakly dominant-strategy profile. Third, the proposition analyzes one natural
accommodation rule rather than optimizing over all rules that may serve the
entrant. Fourth, weak dominance supplies a standard equilibrium profile but
does not imply equilibrium uniqueness. Its finite-portfolio statement also
keeps the controller's bid ceiling fixed when identities are added; it does
not cover identities that manufacture a higher permissible ceiling.

\subsection{Holdback-specific middle-regime sensitivity}
For each $K$, let $\widehat v_{D,K}^R$ be the uniform-capacity activation
threshold at $v_C=1.5$, and set
\[
 v_D=\widehat v_{D,K}^R+\tau(1-\widehat v_{D,K}^R),
 \qquad \tau\in\{0.2,0.4,0.6,0.8\}.
\]
For the power family $F_\theta(a)=a^\theta$, every reported cell satisfies
$\theta K\geq1$, and
\[
 \Psi_{F_\theta}(a)
 =v_Da^{\theta K-1}
  +\left(\frac{v_D}{v_C}\right)^Ka^{\theta K}
\]
is increasing in $a$. Its maximum is therefore attained at $a=1$, so the
activation equation is the same
$v_D+(v_D/v_C)^K=1$ as under uniform capacities. Combining
$K\in\{2,3,5,10\}$ and $\theta\in\{0.5,1,2\}$ gives 48 cells strictly above
activation and below $v_D=1$.

The activation result and the strict pointwise dual loss on the positive
holdback band imply analytically that no holdback-free rule satisfying the
design constraints above attains the class optimum in these cells. This strict ordering does not rely on a
numerical grid. To describe its scale, we evaluate the analytically
characterized optimum by adaptive quadrature and compare it with a
$65{,}536$-node quantile-grid approximation to the monotone-envelope
leader-always-allocate comparator. The activation point is supplied explicitly
to the quadrature so that a thin upper-tail holdback interval is not skipped.
The computed relative differences are positive in every cell, with minimum
$0.000010687\%$, median $0.047907\%$, and maximum $5.571291\%$. These figures
are numerical diagnostics, not certified error bounds for the continuous
leader-always-allocate comparator; in particular, the smallest thick-market values
should be read as scale estimates rather than as the proof of strictness.
The fraction of the remaining first-best gap closed is
\[
 \frac{\PS^*-\PS_{\mathrm{LA}}}{v_C-\PS_{\mathrm{LA}}}.
\]

\begin{table}[t]
\centering
\caption{Grid diagnostic for the holdback-specific gain over the
leader-always-allocate comparator. The class optimum is evaluated by
high-precision quadrature and the comparator on a $65{,}536$-node grid;
percentages summarize 12 $(\theta,\tau)$ cells. Strict positivity follows from
the dual argument, not from this numerical table.}
\small
\begin{tabular}{rrrrr}
\toprule
$K$ & Min approx. gain & Median approx. gain & Max approx. gain & Median gap closed\\
\midrule
2  & 0.0186\% & 0.7991\% & 5.5713\% & 2.2060\%\\
3  & 0.0072\% & 0.2449\% & 2.4283\% & 0.4226\%\\
5  & 0.0011\% & 0.0356\% & 0.3820\% & 0.0406\%\\
10 & 0.000011\% & 0.000318\% & 0.003819\% & 0.000245\%\\
\bottomrule
\end{tabular}
\label{tab:certain-delivery-grid}
\end{table}

The main illustrative point is the $K=2$, $\theta=1$, $\tau=0.8$ cell. Its
leader-always-allocate comparator is also available from the analytic calculation in
the preceding section. It lies near the $v_D=1$ boundary and is chosen to make
the strict separation visible, not to claim that $4.68\%$ is representative
across market thicknesses.

\section{Reproducibility}
The tested reference environment is installed and the complete experiment regenerated
from the package root with
\begin{verbatim}
python -m pip install -r requirements-lock.txt
python code/test_mechanism.py
python code/test_numerics.py
python code/test_accommodation_benchmark.py
python code/verify_full_commitment_normalization.py \
  --nodes 6 \
  --out results/full_commitment_normalization_check.json
python code/experiments.py \
  --out-dir results --figures-dir figures \
  --rank-grid 300 --certain-grid 65536 \
  --mc-draws 1000000 --mc-seed 136
python code/make_figure1.py --out-dir figures
python code/audit_theory_and_misspecification.py \
  --out-dir results --false-name-trials 100000 \
  --seed 20260715
python code/validate_reported_results.py --scope all
\end{verbatim}
The core run recomputes all reported values, figures, and the finite
full-commitment check; unreported diagnostic grids remain in the repository.
The tested environment used Python 3.12.10, NumPy 2.1.3, SciPy 1.18.0,
Matplotlib 3.11.0, and SciPy's HiGHS backend with primal and dual feasibility
tolerances of $10^{-9}$. The reference run used a 64-bit Windows 11 build 26200 system,
an AMD Ryzen AI 9 HX 370 CPU, and 32 GB of memory; no GPU, external data, or
network access is required after dependency installation. The Monte Carlo
check uses NumPy seed 136 and one million draws, and the finite false-name
audit uses Python seed 20260715 and one hundred thousand trials. The saved
Monte Carlo output reports standard errors and normal confidence intervals;
the sensitivity output reports all 48 cells and their minimum, median, and
maximum.

\end{document}